\documentclass[a4paper,reqno,11pt]{amsart}

\usepackage[utf8]{inputenc}
\usepackage[T1]{fontenc}
\usepackage{lmodern}
\usepackage[UKenglish]{babel}
\usepackage[UKenglish]{isodate}
\usepackage{amsmath, amsfonts, amssymb, amsthm, mathtools, cancel} % math packages
\usepackage{geometry} % to customize page layout
\usepackage{mathrsfs} % pour utiliser \mathscr
\usepackage{hyperref} % to produce hypertext links in the document
\hypersetup{colorlinks = true}
\hypersetup{linktocpage} %to only link the page numbers and not the entire table of contents
\usepackage[capitalise]{cleveref} % to allow the format of cross-references to be determined automatically according to the “type” of cross-reference (equation, section, etc.) (capitalise = capitalization of all label names throughout the document)
\usepackage{dsfont}
\usepackage{calligra}
\usepackage{braket}

\theoremstyle{plain}
\newtheorem{thm}{Theorem}[section]

\newtheorem{lem}[thm]{Lemma}

\theoremstyle{definition}

\theoremstyle{remark} 
\newtheorem{rmkx}[thm]{Remark}

\newenvironment{rmk}
  {%
   \pushQED{\qed}\begin{rmkx}}
  {\popQED\end{rmkx}}

\newcommand{\interval}[4]{\mathopen{#1}#2
	\mathclose{}\mathpunct{},#3
	\mathclose{#4}}

\newcommand{\intoc}[2]{\interval{(}{#1}{#2}{]}}
\newcommand{\intco}[2]{\interval{[}{#1}{#2}{)}}
\newcommand{\intoo}[2]{\interval{(}{#1}{#2}{)}}

\newcommand{\bigO}[1]{O\mathopen{}\left(#1\right)}

\newcommand{\abs}[1]{\left\lvert#1\right\rvert}
\newcommand{\norm}[1]{\left\lVert#1\right\rVert}
\newcommand{\scalprod}[2]{\left\langle#1,#2\right\rangle}

\newcommand{\numberset}[1]{\mathbb{#1}}

\newcommand{\R}{\numberset{R}}
\newcommand{\C}{\numberset{C}}

\newcommand{\DD}{\mathcal D} % Dirac+exchange term
\newcommand{\schatten}[1]{\mathfrak{S}_{#1}}
\newcommand{\YY}{\mathcal Y} % space for the Cauchy prob
\newcommand{\CC}{\mathcal C} % Coulomb space
\newcommand{\E}{\mathcal E} % Energy
\newcommand{\HH}{\mathfrak{H}} % hilbert space
\newcommand{\B}{\mathcal B} % bounded operators
\newcommand{\mathbbm}[1]{\text{\usefont{U}{bbm}{m}{n}#1}} % indicator function

\DeclareMathOperator{\tr}{tr}

\DeclareMathOperator{\supp}{supp}
\title[Global well-posedness in a Hartree-Fock model for graphene]{Global well-posedness in a Hartree-Fock model for graphene}
\author[William Borrelli]{William Borrelli}
\address{Politecnico di Milano, Dipartimento di Matematica, P.zza Leonardo da Vinci, 32, 20133 Milano, Italy}
\email{\href{mailto:william.borrelli@polimi.it}
{\nolinkurl{william.borrelli@polimi.it}}}
\author[Umberto Morellini]{Umberto Morellini}
\address{Università di Pisa, Dipartimento di Matematica, Largo Bruno Pontecorvo 5, 56127 Pisa, Italy}
\email{\href{mailto:umberto.morellini@dm.unipi.it}{\nolinkurl{umberto.morellini@dm.unipi.it}}}
\date{8th May 2025}

\begin{document}

\begin{abstract}
    Graphene is a monolayer graphitic film in which electrons behave like two-dimensional Dirac fermions without mass. Its study has attracted a wide interest in the domain of condensed matter physics. In particular, it represents an ideal system to test the comprehension of 2D massless relativistic particles in a laboratory, the Fermi velocity being $300$ times smaller than the speed of light.\\
    In this work, we present a global well-posedness result for graphene in the Hartree-Fock approximation. The model allows to describe the time evolution of graphene in the presence of external time-dependent electric potentials, such as those induced by local charge defects in the monolayer of carbon atoms. Our approach is based on a well established non-perturbative framework originating from the study of three-dimensional quantum electrodynamics.
\end{abstract}

\maketitle

\tableofcontents

\section{Introduction}
Graphene is a monolayer of carbon atoms packed into a dense honeycomb crystal structure. It is a two-dimensional allotrope of carbon whose comprehension can be considered as a first step for the understanding of electronic properties of other allotropes. Indeed, graphene can be seen as either an individual atomic plane extracted from graphite or an unrolled single-wall carbon nanotubes or a giant flat fullerene molecule. All these carbon-based systems show different structures, originating from the flexibility of their chemical bonds, with different physical properties depending mainly on the dimensionality of the structures themselves. Moreover, in condensed matter physics, most of the systems are described by the non-relativistic Schrödinger equation and their low energy excitations are quasi-particles with finite effective mass. On the contrary, electrons in graphene behave like two-dimensional Dirac fermions without mass. Indeed, this carbon allotrope has a quasi-particle dispersion relation which is conical at the degeneracy points, leading to an effective massless Dirac equation, where the speed of light is replaced by the Fermi velocity $v_\mathrm{F}$. These massless quasi-particles are proposed to be responsible for various anomalous phenomena observed in graphene \cite{GonGuiVoz-1996-PRL, NovGeiMorJia-2005-Nat, ZhaTanStoKim-2005-Nat}. For all these reasons, it has attracted a huge interest in the last decades \cite{CasGuiPerNov-2009-RMP}.\par
The main purpose of this paper is to present a global well-posedness result for the Cauchy problem describing the time evolution of graphene in the presence of external time-dependent electric fields, induced for instance by a set of local defects in the system. The Fermi velocity in graphene being small (corresponding to a large value of the coupling constant), perturbation theory cannot be applied to this material. For this reason, we need to rely on non-perturbative methods such as mean-field approximations of Hartree-Fock type. As a consequence, this work is deeply inspired by \cite{HaiLewSpa-2012-JMP}, where the authors study the Hartree-Fock approximation of graphene in an infinite volume with instantaneous Coulomb interactions, proving the existence of the ground state and deriving its properties. This paper, in turn, represents an extension to 2D massless electrons of \cite{HaiLewSer-2005-CMP,HaiLewSer-2005-JPA,HaiLewSol-2007-CPAM}, where the authors study the Hartree-Fock approximation of three-dimensional quantum electrodynamics (QED). This model, named after Bogoliubov-Dirac-Fock (BDF), was suggested by a seminal paper of Chaix and Iracane \cite{ChaIra-1989-JPB} and we refer the reader to \cite{EstLewSer-2008-BAMS,HaiLewSerSol-2007-PRA} for a complete review. The BDF framework is completely non-perturbative and holds for all values of the coupling constant $0\leq\alpha\leq 4/\pi\simeq 1.27$. For $\alpha>4/\pi$, the model is known to become unstable \cite{ChaIraLio-1989-JPB}.\par
While in the physics literature graphene is often studied within alternative frameworks such as Density Functional Theory (DFT), to our knowledge a rigorous mathematical foundation for its time-dependent extension (TDDFT) remains largely undeveloped \cite{FouLamLewSor-2016-PRA}. For this reason, we believe that time-dependent mean-field models continue to deserve mathematical attention, particularly in the context of graphene due to its intrinsically non-perturbative nature.\par
Our two-dimensional theory contains ultraviolet (UV) divergences which represent a well known issue in quantum electrodynamics. Therefore, one needs to apply one of the well established \textit{regularization} techniques (see for instance \cite{Lei-1975-RMP,PauVil-1949-RMP}) in order to overcome difficulties in the high frequencies domain. It is known that these methods introduce some non-physical parameters which make unclear whether the obtained results are universal or depend on the chosen regularization technique. As a consequence, several procedures, known under the name of \textit{renormalization}, have been studied in order to absorb the regularization parameters and get again a physical theory. In the present work, we impose a sharp ultraviolet cut-off $\Lambda>0$, as done for example in the three-dimensional electrostatic case \cite{HaiLewSer-2005-CMP,HaiLewSer-2005-JPA,HaiLewSol-2007-CPAM}. Because of the underlying lattice structure, this choice of the regularization method is natural for graphene.  For this reason, we do not need to discuss the renormalization issue in our model.

We mention that graphene has also been extensively studied in a series of papers, such as \cite{GiuMas-2009-PRB,GiuMas-2009-CMP,GiuMasPor-2010-PRB,GiuMasPor-2012-AP} and related references, where rigorous lattice models are considered. Moreover, we refer the reader to \cite{FefWei-2012-JAMS,FefWei-2013-CMP} for the properties of honeycomb Schr\"odinger operators modeling graphene and similar materials and to \cite{ArbSpa-2018-JMP, Bor-2017-JDE, Bor-2018-CVPDE, Bor-2021-CCM,BorFra-2019-TAMS} for the study of effective nonlinear Dirac equations.
\par
%%%%%%%%%%%%%%%%%%%%%%%%%%%%%%%%%%%%%%%%%%%%%%%
\noindent{\subsection*{Organization of the paper}}
The paper is organized as follows. In \cref{sec:hartree-fock-theory-of-graphene}, the Hartree-Fock theory of graphene is introduced and discussed in detail, following \cite{HaiLewSpa-2012-JMP}. It represents an extension to the two-dimensional case of the Bogoliubov-Dirac-Fock model, where the Dirac sea is replaced by the Fermi sea and ground state properties of graphene are studied in the presence of an external electrostatic potential. In the present paper, we study the Cauchy problem for the corresponding BDF-type evolution equation in the presence of an external time-dependent charge density, representing local charge defects in the monolayer of carbon atoms. The two-dimensional electrons interact via the instantaneous Coulomb potential, with kinetic energy given by the massless Dirac operator. Our main result on the global well-posedness of the above-mentioned Cauchy problem, is proved in \cref{sec:global-well-posedness}. The result holds under the sole assumption that
\[
0\leq\alpha\leq\alpha_\mathrm{c}\quad\text{or, equivalently,}\quad v_\mathrm{F}\geq v_\mathrm{c},
\]
where the critical effective coupling constant $\alpha_\mathrm{c}$ (defined below) is such that $\alpha_\mathrm{c}>0.48637$, or
equivalently $v_\mathrm{c}<2.0560$. These estimates are nonetheless quite rough. Numerical computations suggest a critical velocity $v_\mathrm{c}$ of approximately $0.36$, corresponding to a critical coupling $\alpha_\mathrm{c}$ of around $2.7$. This includes the case of graphene, for which the Fermi velocity is approximately $v_\mathrm{F}\simeq 1.1$ \cite{HaiLewSpa-2012-JMP}. Here the latter plays the role of the speed of light in the genuine relativistic case, while $\alpha$ is the coupling constant for the interaction.\par
Let us finally mention that our main result can be regarded as a preliminary step towards a more detailed analysis of the dynamics such as, for instance, scattering properties, which we intend to explore in future work.
%%%%%%%%%%%%%%%%%%%%%%%%%%%%%%%%%%%%%%%%%%%%%%%%%%%%%%%%%%%%
%%%%%%%%%%%%%%%%%%%%%%%%%%%%%%%%%%%%%%%%%%%%%%%%%%%%%%%%%%%%

\section{Hartree-Fock theory of graphene and main result}\label{sec:hartree-fock-theory-of-graphene}
\subsection{Derivation of the Hartree-Fock energy}
We start by explaining how our model is derived from two-dimensional quantum electrodynamics, following \cite{HaiLewSpa-2012-JMP}. For the convenience of the reader, we then give its main properties, referring to that paper for more details and proofs. We also adopt notations from that reference. Notice that for the moment the external electric potential will be assumed to be stationary.

Electrons in graphene are essentially confined in 2D, but subjected to an electric field which acts in three spatial dimensions. For this reason we assume the interaction to be given by the 3D instantaneous Coulomb potential. Our starting point is the following \emph{formal} 2D QED-type Hamiltonian, given in the Coulomb gauge by \cite[Eq.~(218)]{CasGuiPerNov-2009-RMP}
\begin{equation}\label{eq:HQED}
    \mathbb H^\varphi=v_\mathrm{F}\int_{\R^2}\Psi^*\left(x\right)\,\boldsymbol{\sigma}\cdot\left(-i\nabla\right)\Psi\left(x\right)\,dx+\int_{\R^2}\varphi\left(x\right)\rho\left(x\right)\,dx+\frac{1}{2}\iint_{\R^2\times\R^2}\frac{\rho\left(x\right)\rho\left(y\right)}{\abs{x-y}}\,dxdy\,.
\end{equation}
Notice that in our model photons are neglected.
In \eqref{eq:HQED},  $\boldsymbol{\sigma}=\left(\sigma_1,\sigma_2\right)$, and $\sigma_1,\sigma_2$ are the first two Pauli matrices. Here, and in what follows, we shall denote by
\begin{equation}
    D^0=-i\boldsymbol{\sigma}\cdot\nabla=-i\sigma_1\partial_{x_1}-i\sigma_1\partial_{x_2}
\end{equation}
the massless 2D Dirac operator. The Hamiltonian acts on the fermionic Fock space $\mathscr{F}$. Notice that we work in atomic units with Planck’s constant and the mass of the electrons normalized to $1$, $\hbar=m=1$. We also assume that $e^2/\left(4\pi\kappa\right)=1$ where $\kappa$ is the dielectric screening constant of the substrate on which our graphene sheet is placed \cite{CasGuiPerNov-2009-RMP,Mis-2008-EPL}. In these units, the (bare) Fermi velocity $v_\mathrm{F}$ is about $1.1$. Thus, the electrostatic interaction, that is the last term of \eqref{eq:HQED}, and the kinetic part are of the same order. Moreover, since the effective speed of sound in graphene is much smaller than the speed of light, the use of instantaneous interactions is justified. We point out that we do not consider the proper spin of electrons, since we do not include a magnetic field, but rather their \emph{pseudo-spin}. The latter accounts for the presence of the two triangular lattices that compose the hexagonal structure of the graphene lattice. The Hamiltonian well describes the electrons with energies about the Fermi level, assuming they live on scale much larger that the lattice size, so that the continuous approximation leading to \eqref{eq:HQED} is justified. 

In \eqref{eq:HQED}, $\Psi(x)$ is the second-quantized field operator annihilating an electron at $x$ and satisfying the anti-commutation relation
\begin{equation}\label{eq:acomm}
    \Psi^*\left(x\right)_\sigma\Psi\left(y\right)_\nu+\Psi\left(y\right)_\nu\Psi^*\left(x\right)_\sigma=2\delta_{\sigma,\nu}\delta\left(x-y\right)\,,
\end{equation}
$\sigma,\nu\in\{\pm1/2\}$ being the pseudo-spin variables. The density operator is defined as
\begin{equation}\label{eq:densop}
\rho\left(x\right)=\sum^2_{\sigma=1}\frac{\left[\Psi^*_\sigma\left(x\right),\Psi_\sigma\left(x\right)\right]}{2}\,,
\end{equation}
where $\left[a,b\right]=ab-ba$. The function $\varphi$ denotes a local external electrostatic potential, for instance created by some defect in the system.

We emphasize that \eqref{eq:HQED} does not contain any normal ordering or notion of (bare) electrons and positrons. The commutator used in the definition \eqref{eq:densop} of $\rho$ is a kind of renormalization, independent of any reference. In the three-dimensional case, it has been introduced by Heisenberg \cite{Hei-1934-ZF} and has been widely used by Schwinger \cite{Sch-1948-PR, Sch-1949-PR,Sch-1951-PR} as a necessity for a covariant formulation of QED. Indeed, in this way, the Hamiltonian $\mathbb H^\varphi$ has the property of being charge conjugation invariant, as we have
\[
\mathscr{C}\rho(x)\mathscr{C}^{-1}=-\rho(x)\,,\qquad \mathscr{C}\mathbb H^\varphi\mathscr{C}^{-1}=\mathbb H^{-\varphi}\,,
\]
where $\mathscr{C}$ is the charge conjugation operator on $\mathscr{F}$ \cite[Formula~(1.81)]{Tha-1992-book}. Observe that the Hamiltonian \eqref{eq:HQED} is not well defined as a self-adjoint operator on $\mathscr{F}$ and, even formally, it is not bounded from below. However, it is possible to define it properly in a box with periodic boundary conditions, imposing an ultraviolet cut-off, as done in \cite{HaiLewSol-2007-CPAM} for the 3D case.

As a further assumption, we consider a mean-field approximation by restricting ourselves to Hartree-Fock states. These states $\ket \Omega$ are completely determined by the corresponding \emph{one-body density matrix}, that is an operator acting on the one-body Hilbert space $L^2\left(\R^2,\C^2\right)$ and whose integral kernel is defined as
\[
\gamma\left(x,y\right)_{\sigma,\sigma'}=\left\langle\Omega\left\vert \Psi^*\left(x\right)_\sigma\Psi\left(y\right)_{\sigma'}\right\vert\Omega\right\rangle\,,
\]
Because of the fermionic nature of the system, this operator satisfies $0\leq\gamma\leq 1$. In the usual Hartree-Fock theory, the energy depends only on the orthogonal projection $\gamma$. However, due to the definition given for the density operator $\rho$ in \eqref{eq:densop}, the energy in QED is more easily expressed in terms of the \emph{renormalized density matrix}
\[
\gamma_{\mathrm{ren}}\left(x,y\right)_{\sigma,\sigma'}=\left\langle \Omega \left\vert\frac{\left[\Psi^+\left(x\right)_\sigma,\Psi\left(y\right)_{\sigma'}\right]}{2} \right\vert\Omega\right\rangle\,,
\]
or equivalently
\[
\gamma_{\mathrm{ren}}=\gamma - \frac{I}{2}\,,
\]
where $I$ is the identity operator on the one-body Hilbert space. Computing the average value of the energy over a Hartree-Fock state, one finds
\begin{equation}\label{eq:HEV}
\langle \mathbb H^\varphi\rangle = \mathcal{E}^\varphi_{\mathrm{HF}}(\gamma-I/2)+C\,,
\end{equation}
where the constant $C$ is divergent in infinite volume (see \cite{HaiLewSpa-2012-JMP} which refers to \cite[Equation~(2.18)]{HaiLewSol-2007-CPAM}) and the first term is the well known Hartree-Fock energy (evaluated in the renormalized density matrix $\gamma_{\mathrm{ren}}$, as previously disclosed)
\begin{multline}\label{eq:EHF}
    \mathcal{E}^\varphi_{\mathrm{HF}}(\gamma_{\mathrm{ren}})=v_\mathrm{F}\tr(D^0\gamma_{\mathrm{ren}})+\int_{\R^2}\varphi(x)\rho_{\gamma_{\mathrm{ren}}}(x)\,dx \\
    +\frac{1}{2}\iint_{\R^2\times\R^2}\frac{\rho_{\gamma_{\mathrm{ren}}}(x)\rho_{\gamma_{\mathrm{ren}}}(y)}{\vert x-y\vert}\,dxdy-\frac{1}{2}\iint_{\R^2\times\R^2}\frac{\vert \gamma_{\mathrm{ren}}(x,y)\vert^2}{\vert x-y\vert}\,dxdy\,.
\end{multline}
In \eqref{eq:EHF}, $\rho_{\gamma_{\mathrm{ren}}}$ is the density associated to $\gamma_{\mathrm{ren}}$. As usual, the first term in \eqref{eq:EHF} is the kinetic energy, while the second one is the interaction energy between the electrons and the external potential. The last two terms are, respectively, the \emph{direct term}, describing the interactions between electrons, and the \emph{exchange term}, that is a purely quantum correction of the error in the self-interactions between electrons described by the direct term. Observe that, even discarding the constant in \eqref{eq:HEV}, the energy \eqref{eq:EHF} is not well defined yet. Indeed, the kinetic term is ill-posed as the Dirac operator is unbounded and the states considered always have infinitely many particles. Moreover, the density of charge needs to be properly defined as well. These issues will be addressed in the next section.
%%%%%%%%%%%%%%%%%%%%%%%%%%%%%%%%%%%%%%%
\subsection{Functional setting and definition of the energy}
As already remarked, the lattice structure of graphene naturally imposes an ultraviolet cut-off $\Lambda$, proportional to the reciprocal of the lattice spacing. Its physical value is $\Lambda\simeq 0.1 \mathring{\mathrm A}^{-1}$. For this reason, we replace the one-particle Hilbert space $L^2(\R^2,\C^2)$ by
\begin{equation}\label{eq:HL}
\mathfrak{H}_\Lambda=\{\psi\in L^2\left(\R^2\,,\C^2\right)\,:\,\supp\left(\widehat\psi\right)\subseteq B(0,\Lambda)\}\,,
\end{equation}
the space of $L^2$ functions with Fourier transform supported in the ball $B\left(0,\Lambda\right)$.

In order to properly define the energy and the associated (unique) ground state, we first consider the case without external potential, as done in \cite{HaiLewSpa-2012-JMP}. Let us fix then $\varphi\equiv 0$. Neglecting the interaction between electrons and considering the system confined to a box with periodic boundary conditions, the unique minimizer is obviously given by the free Dirac sea. In the thermodynamic limit, it converges to the infinite-volume free Dirac sea, given by a Hartree-Fock state containing all the negative energy electrons. Its density matrix and renormalized density matrix are, respectively, given by
\begin{equation}\label{eq:freesee}
\gamma^0=P^0_-=\mathds1\left(D^0\leq 0\right)\,,\quad\mbox{and}\quad \gamma^0_{\mathrm{ren}}=P^0_- -\frac{I}{2}=-\frac{D^0}{2\vert D^0\vert} \,,
\end{equation}
where $P^0$ is the projector onto the negative spectral subspace of $D^0$. As proved in \cite{HaiLewSol-2007-CPAM} in the 3D case, when the interaction terms are switched on, the free Dirac sea changes but is still invariant by translation. The same result is true also in our case, as shown in \cite{HaiLewSpa-2012-JMP}. The energy per unit volume of a translation-invariant state $\gamma_{\mathrm{ren}}=f\left(p\right)$ is given by
\begin{equation}\label{eq:EUV}
    \mathcal F\left(\gamma_{\mathrm{ren}}\right)=\frac{1}{\left(2\pi\right)^2}\left(v_\mathrm{F} \int_{B\left(0,\Lambda\right)}\tr_{\C^2}\left(\boldsymbol{\sigma}\cdot p \,f_{\mathrm{ren}}\left(p\right)\right)\,dp-\frac{1}{2}\int_{\R^2}\frac{\check{f}_{\mathrm{ren}}}{\vert x\vert}\,dx\right)\,.
\end{equation}
Notice that in Fourier space the state $\gamma_{\mathrm{ren}}$ is given by the multiplication by a $2\times2$ matrix $f_{\mathrm{ren}}\left(p\right)$. We have assumed that such a state has a vanishing density of charge, $\tr_{\C^2} f\left(p\right)=0$ for any $p\in B\left(0,\Lambda\right)$, so that the direct term in \eqref{eq:EUV} is absent. Indeed, the charge density must be constant in position space, as
\[
\rho_{\gamma_{\mathrm{ren}}}\left(x\right)=\tr_{\C^2}\gamma_{\mathrm{ren}}\left(x,x\right)=(2\pi)^{-1}\tr_{\C^2}\check{f}_{\mathrm{ren}}\left(0\right)=\left(2\pi\right)^{-1}\int_{B\left(0,\Lambda\right)}\tr_{\C^2}f\left(p\right)dp\,,
\]
but the Coulomb energy of a constant density is not proportional to the volume unless it vanishes. This is the case for the minimizer of \eqref{eq:EUV}. The renormalized density matrix $\gamma^0_{\mathrm{ren}}$ in Fourier space reads
\[
f^0_{\mathrm{ren}}\left(p\right)=-\frac{\boldsymbol{\sigma}\cdot p}{2\vert p\vert}=\mathds{1}_{\intoo{-\infty}{0}}\left(\boldsymbol{\sigma}\cdot p\right)- \frac{I_{\C^2}}{2}\,,
\]
so that its charge density vanishes as remarked. Computing the derivative of the energy \eqref{eq:EUV} at $f^0_{\mathrm{ren}}(p)$ one gets the \emph{mean-field operator} of this state
\begin{equation}\label{eq:MFD}
\DD^0=v_\mathrm{F} D^0-\frac{\check{f}^0_{\mathrm{ren}}\left(x-y\right)}{2\pi\abs{x-y}}\,,
\end{equation}
where $-(2\pi)^{-1}\check{f}^0_{\mathrm{ren}}(x-y)$ is the kernel in position space of the inverse Fourier transform of $f^0_{\mathrm{ren}}$. Thus, \eqref{eq:MFD} is given by the free Dirac operator together with the exchange term.
\begin{lem}[{\cite[Lemma~3.2]{HaiLewSpa-2012-JMP}}]
The mean-field operator \eqref{eq:MFD} can be written, in Fourier space, as
\begin{equation}\label{eq:effd}
    \widehat{\DD}^0\left(p\right)=v_{\mathrm{eff}}\left(p\right)\boldsymbol{\sigma}\cdot p\,,
\end{equation}
where
\[
v_{\mathrm{eff}}:=v_{\mathrm F}+g\left(\Lambda/\vert p\vert\right)
\]
and
\begin{equation}\label{eq:g}
g\left(R\right)=\frac{1}{2\pi}\int^\pi_0\int^R_0\frac{\cos\theta}{\sqrt{r^2-2r\cos\theta+1}}\,rdrd\theta\,.
\end{equation}
The function $g$ is increasing on $\intco{1}{+\infty}$ and satisfies
\[
g\left(1\right)\simeq 0.1234\,.
\]
\end{lem}
\begin{rmk}
By \eqref{eq:effd}, using that $g\left(\Lambda/\vert p\vert\right)\geq g\left(1\right)$ we get
\begin{equation}\label{eq:dp}
    \abs{\widehat{\DD}^0\left(p\right)}\geq\left(v_{\mathrm F}+g\left(1\right)\right)\abs{\widehat{D}^0\left(p\right)}=\left(v_{\mathrm F}+g\left(1\right)\right)\abs{p}\,.
\end{equation}
\end{rmk}
We see that the effective operator is of Dirac type with momentum-dependent Fermi velocity, which is logarithmically divergent at $p=0$:
\[
v_{\mathrm{eff}}\left(p\right)=\frac{1}{4}\log\frac{\Lambda}{\vert p\vert}+\bigO{1}\,,\qquad \mbox{as $\vert p\vert\to 0$\,.}
\]
This phenomenon is known in the literature as \emph{Kohn anomaly} \cite{CasGuiPerNov-2009-RMP}. Here, the $\bigO{1}$ is independent of $\Lambda$. One remarkable consequence of the above Lemma is that 
\begin{equation}\label{eq:ids}
\gamma^0_{\mathrm{ren}}=-\frac{\DD^0}{2\abs{\DD^0}}=\mathds{1}_{\intoo{-\infty}{0}}\left(\DD^0\right)-\frac{I}{2}\,,
\end{equation}
as, in momentum space, $\DD^0$ is equal to the free operator $D^0$ multiplied by a positive function. This means that the \emph{non-interacting} free Dirac sea solves the same non-linear equation as the interacting one, in contrast with the three-dimensional massive case \cite{HaiLewSol-2007-CPAM,LieSie-2000-CMP}.\par
Let us now introduce the following function:
\begin{equation}\label{eq:h-function}
    h\left(v_\mathrm{F}\right)=\sup_{\varphi\in\mathfrak{H}_1\setminus\{0\}}\frac{\int_{\R^2}\frac{\abs{\varphi\left(x\right)}^2}{\abs{x}}dx}{\scalprod{\varphi}{\abs{p}\left(v_\mathrm{F}+g\left(1/\abs{p}\right)\right)\varphi}}\,,
\end{equation}
where $g$ is defined in \eqref{eq:g}.
\begin{thm}[{\cite[Theorem~1]{HaiLewSpa-2012-JMP}}]\label{thm:gsfg}
Let $\Lambda>0$. If
\begin{equation}\label{eq:vflb}
v_F\geq v_\mathrm{c}=h^{-1}\left(2\right)\,,
\end{equation}
then the minimization problem
\[
\min\{\mathcal F \left(\gamma_{\mathrm{ren}}\right)\,:\,\gamma_{\mathrm{ren}}=f_{\mathrm{ren}}\left(p\right)\,,\,-\frac{\mathds{1}_{B\left(0,\Lambda\right)}}{2}\leq f_{\mathrm{ren}}\leq \frac{\mathds{1}_{B\left(0,\Lambda\right)}}{2}\,,\,\rho_{\gamma_{\mathrm{ren}}}\equiv 0 \}\,,
\]
with $\mathcal F$ defined in \eqref{eq:EUV}, admits the unique minimizer \eqref{eq:ids}.
\end{thm}
As observed in \cite{HaiLewSpa-2012-JMP}, numerical
calculations suggest that $v_\mathrm{c}$ in \eqref{eq:vflb} is about $0.36$, which covers the case of graphene, for which $v_\mathrm{F}\simeq 1.1$. The arguments in \cite{HaiLewSol-2007-CPAM} allows to show that \eqref{eq:ids} is also the thermodynamic limit of ground states of the Hartree-Fock energy, without assuming invariance by translation.
\medskip

So far we have discussed the case without external potential. Let us now assume that $\varphi\neq 0$. More precise assumptions will be given later. 

In this case, we consider again Hartree-Fock states, described by their density matrix $\gamma$ or, equivalently, by their renormalized density matrix $\gamma_{\mathrm{ren}}=\gamma-I/2$. As before, their energy \eqref{eq:HEV} is infinite. In order to get a well defined quantity, we can, at least formally, subtract the infinite constant $\E^\varphi_{\mathrm{HF}}(\gamma^0_{\mathrm{ren}})=\E^0_{\mathrm{HF}}(\gamma^0_{\mathrm{ren}})$, representing the energy of the vacuum. Once again, a rigorous justification is possible through a thermodynamic limit, arguing as in \cite{HaiLewSol-2007-CPAM}. Then, a (formal) computation gives
\[
\E^\varphi_{\mathrm{HF}}\left(\gamma_{\mathrm{ren}}\right)-\E^\varphi_{\mathrm{HF}}\left(\gamma^0_{\mathrm{ren}}\right)=\E^\varphi_{\mathrm{BDF}}\left(Q\right)\,,
\]
where
\[
Q=\gamma_{\mathrm{ren}}-\gamma^0_{\mathrm{ren}}=\gamma-P^0_-\,,
\]
and 
\begin{multline}\label{eq:bdfen}
    \E^\varphi_{\mathrm{BDF}}\left(Q\right)=\tr\left(\DD^0Q\right)+\int_{\R^2}\varphi\left(x\right)\rho_Q\left(x\right)\,dx \\
    +\frac{1}{2}\iint_{\R^2\times\R^2}\frac{\rho_Q\left(x\right)\rho_Q\left(y\right)}{\vert x-y\vert}\,dxdy -\frac{1}{2}\iint_{\R^2\times\R^2}\frac{\vert Q\left(x,y\right)\vert^2}{\vert x-y\vert}\,dxdy\,
\end{multline}
is the so-called \emph{Bogoliubov-Dirac-Fock} energy. It has the same form of the usual Hartree-Fock energy, but now $D^0$ is replaced by $\DD^0$, defined in \eqref{eq:MFD}, and $Q$ is not a density matrix itself, but rather the difference of two density matrices. Observe that $Q$ trivially verifies the property
\begin{equation}\label{eq:qineq}
-P^0_-\leq Q\leq 1-P^0_-=P^0_+\,.
\end{equation}
\begin{rmk}
    Subtracting the infinite energy of the free Fermi sea in order to get a bounded-below energy is an approach based on \cite{HaiLewSpa-2012-JMP}, where this Hartree-Fock type model for graphene is studied. It generalizes an idea used for the first time in \cite{HaiLewSol-2007-CPAM} in the study of the no-photon Hartree-Fock approximation of QED in the purely electrostatic case (thus, with the Fermi sea replaced by the Dirac sea). Later on, this idea has also been employed in the electromagnetic case in \cite{GraHaiLewSer-2013-ARMA} and in the positive temperature case in \cite{HaiLewSei-2008-RMP,Mor-2025-AHP}. In previous works always dealing with no-photon Hartree-Fock approximation of QED \cite{ChaIra-1989-JPB,HaiLewSer-2005-CMP}, another point of view based on the notion of normal ordering was given.
\end{rmk}
Let us now describe a functional framework that allows to properly define the energy functional \eqref{eq:bdfen}. As remarked in \cite{HaiLewSpa-2012-JMP}, a state has a finite relative kinetic energy if the latter is understood as
\begin{equation}\label{eq:trace}
\tr\left(\DD^0Q\right):=\tr\left(\abs{\DD^0}^{1/2}\left(Q^{++}-Q^{--}\right)\abs{\DD^0}^{1/2}\right)\,,
\end{equation}
assuming that
\[
\abs{\DD^0}^{1/2}Q^{\pm\pm}\abs{\DD^0}^{1/2}\in\mathfrak{S}^1\left(\mathfrak H_\Lambda\right)\,.
\]
Concerning the off-diagonal terms, this only allows to infer the following property
\[
Q^{\pm\mp}\abs{\DD^0}^{1/2}\in\mathfrak{S}^2\left(\mathfrak H_\Lambda\right)\,.
\]
Here, $\mathfrak{S}^1$ and $\mathfrak{S}^2$ are the Schatten ideals of trace-class and Hilbert-Schmidt operators, respectively, and
\[
Q^{\epsilon,\epsilon'}=P^0_{\epsilon} Q P^0_{\epsilon'}\,,\qquad \epsilon,\epsilon'\in\{+,-\}\,.
\]

The discussion above leads to consider the following space of operators
\begin{equation}\label{eq:X}
\widetilde{\mathcal X}=\{Q\in\mathcal B\left(\mathfrak{H}_\Lambda\right)\,:\, Q^*=Q\,,\abs{\DD^0}^{1/2}Q^{\pm\pm}\abs{\DD^0}^{1/2}\in\mathfrak{S}^1\,, Q^{\pm\mp}\abs{\DD^0}^{1/2}\in\mathfrak{S}^2\}.
\end{equation}
Nevertheless, we also need to well define the density $\rho_Q$ as well as the direct term in \eqref{eq:bdfen}. To this aim, observe that the natural space for the charge density is the so-called \emph{Coulomb space}
\[
\mathcal C=\left\{\varphi\,:\, D\left(\varphi,\varphi\right)=2\pi\int_{\R^2}\frac{\abs{\widehat{\varphi}\left(p\right)}^2}{\vert p\vert} \,dp<\infty\right\}\,.
\]
The quantity $D\left(\varphi,\varphi\right)\geq 0$ is the Coulomb energy of a charge density. By \cite[Lemma 4.5]{HaiLewSpa-2012-JMP}, if $Q\in\widetilde{X}$ and $-P^0_-\leq Q\leq P^0_+$, the charge density belongs to $\mathcal C$, $\rho_Q\in\mathcal C$, and then the direct term in \eqref{eq:bdfen} is finite. Under the same hypothesis and assuming \eqref{eq:vflb}, also the exchange term is well defined. Indeed, as also shown in the proof of \cref{stat:global-existence} below, Kato's inequality allows the kinetic term to control the exchange one. This motivates assumption \eqref{eq:vflb}, also in the statement of \cref{thm:gsfg}. The above discussion can be summarized in the following result.
\begin{lem}[{\cite[Lemma~4.6]{HaiLewSpa-2012-JMP}}]\label{lem:freebdf}
    Let $\Lambda>0$. Then the map $Q\mapsto \E^0_{\mathrm{BDF}}\left(Q\right)$ is well defined and continuous on $\widetilde{\mathcal X}$. If $v_{\mathrm F}$ satisfies
    \begin{equation*}
    v_{\mathrm F}\geq v_\mathrm{c}=h^{-1}\left(2\right)\,,
    \end{equation*}
    then one has
    \[
    0\leq \E^0_{\mathrm{BDF}}\left(Q\right)<\infty\,,
    \]
    for all $-P^0_-\leq Q\leq P^0_+$ such that $\abs{\DD^0}^{1/2}Q^{\pm\pm}\abs{\DD^0}^{1/2}\mathfrak{S}^1$. Moreover, $\E^0_{\mathrm{BDF}}\left(Q\right)=0$ if and only if $Q=0$.
\end{lem}
The above lemma states that the BDF energy is well defined and positive definite when $\varphi=0$. The external electrostatic potential, corresponding for instance to a local defect in the monolayer of carbon atoms, can be modeled by a potential of the form
\[
\varphi=-\nu\ast\frac{1}{\vert x\vert}\,,
\]
with $\nu$ the associated density of charge. Thus, the BDF energy can be rewritten as
\[
\E^\varphi_{\mathrm{BDF}}\left(Q\right)=\E^0_{\mathrm{BDF}}\left(Q\right)-D\left(\rho_Q,\nu\right)\,,
\]
where the last term is given by the scalar product corresponding to the Coulomb energy. If $v_{\mathrm F}$ satisfies \eqref{eq:vflb}, then the positivity of $\E^0_{\mathrm{BDF}}$ gives
\begin{equation}\label{eq:energy-lower-bound}
    \E^\varphi_{\mathrm{BDF}}\left(Q\right)\geq-\frac{1}{2}D\left(\nu,\nu\right)\,,
\end{equation}
which means that the BDF energy is bounded from below if $\nu\in\mathcal C$. Thus, we can consider the following minimization problem
\begin{multline}\label{eq:bdfmin}
    E^\varphi_{\mathrm{BDF}}=\inf\{\E^\varphi_{\mathrm{BDF}}\left(Q\right)\,:\, Q\in\mathcal B\left(\mathfrak H_\Lambda\right)\,,\\
    -P^0_-\leq Q=Q^*\leq P^0_+\,, \abs{\DD^0}^{1/2} Q^{\pm\pm}\abs{\DD^0}^{1/2}\in\mathfrak{S}^1\}\,.
\end{multline}
The following result from \cite{HaiLewSpa-2012-JMP} proves the existence of a ground state of interacting graphene in the presence of defects.
\begin{thm}[{\cite[Theorem~2]{HaiLewSpa-2012-JMP}}]\label{stat:existence-minimizer-ext-potential}
    Let $\Lambda>0$ be a fixed UV cut-off and let $v_{\mathrm{F}}$ satisfy \eqref{eq:vflb}. For any $\nu\in \mathcal C$, the problem \eqref{eq:bdfmin} admits at least one minimiser, which also satisfies the following self-consistent equation
    \begin{equation}\label{eq:scQ}
        Q+P^0_-=\mathds{1}_{\mathcal I}\left(\DD^0-\nu\ast\frac{1}{\vert x\vert}+\rho_Q\ast\frac{1}{\vert x\vert}- \frac{Q(x,y)}{\vert x-y\vert} \right)\,,
    \end{equation}
    with $\mathcal I=\intoo{-\infty}{0}$ or $\mathcal I=\intoc{-\infty}{0}$. In terms of the corresponding density matrix $\gamma=Q+P^0_-$, this rewrites as
    \begin{equation}\label{eq:gammaop}
        \gamma=\mathds{1}_{\mathcal{I}}\left(v_{\mathrm F}\boldsymbol{\sigma}\cdot\left(-i\nabla\right)-\nu\ast\frac{1}{\vert x\vert}+\rho_{\gamma-I/2}\ast\frac{1}{\vert x\vert}-\frac{(\gamma-I/2)(x,y)}{\vert x-y\vert} \right)\,.
    \end{equation}
\end{thm}
The density matrix $\gamma$ in \eqref{eq:gammaop} is the projector onto the negative spectral subspace of the mean-field operator
\begin{equation}\label{eq:mfdv}
\DD_Q=v_{\mathrm F}\boldsymbol{\sigma}\cdot\left(-i\nabla\right)-\nu\ast\frac{1}{\vert x\vert}+\rho_{\gamma-I/2}\ast\frac{1}{\vert x\vert}-\frac{\left(\gamma-I/2\right)\left(x,y\right)}{\vert x-y\vert}\,.
\end{equation}
It takes into account both the external potential and the self-consistent one created by the state $Q$ itself.

Our goal in this paper is to study the equation describing the time evolution of such an infinite rank projection $\gamma\left(t\right)$. This equation is a generalization of the BDF evolution equation to the two-dimensional case of graphene, and formally reads
\begin{equation}\label{eq:evolution-equation}
    i\frac{d}{dt}\gamma\left(t\right)=\left[\DD_{Q\left(t\right)}, \gamma\left(t\right)\right],
\end{equation}
where $Q\left(t\right)=\gamma\left(t\right)-P^0_-$ and $\DD_{Q\left(t\right)}$ is the mean-field operator defined above. 

Notice that the minimizer $\overline Q$ of the BDF energy given by \cref{stat:existence-minimizer-ext-potential} is a stationary solution of \eqref{eq:evolution-equation}, as it verifies
\[
\left[D_{\overline{Q}}, \overline{\gamma}\right]=0.
\]
\cref{eq:evolution-equation} is similar to the well known time-dependent Hartree-Fock model of non-relativistic quantum mechanics \cite{CanLeB-1999-MMMAS,ChaGla-1975-JMP}. In the relativistic framework, an equation similar to \eqref{eq:evolution-equation} was first introduced by Dirac in \cite{Dir-1934-MPCPS,Dir-1934-SR}, neglecting the exchange term which is now well established in mean-field theories for fermions. \cref{eq:evolution-equation} can be seen as the natural evolution equation of the BDF model for graphene, since the right-hand side of \eqref{eq:evolution-equation} is the formal derivative of the energy  $\E^\varphi_{\mathrm{BDF}}$ \cite{ChaIra-1989-JPB,HaiLewSpa-2005-LMP}.

Looking at the dynamical problem \eqref{eq:evolution-equation}, hereafter we are going to consider a time-dependent charge density $\nu=\nu(t)$ such that
\begin{equation}\label{eq:n(t)}
\nu\in C^1\left(\intco{0}{+\infty},\CC\right).
\end{equation}
Thus, our goal is to study the following system
\begin{equation}\label{eq:main-cauchy-problem}
    \begin{cases}
        i\frac{d}{dt}\gamma\left(t\right)=\left[\DD_{Q\left(t\right)}, \gamma\left(t\right)\right]\,, \\[6pt]
        \gamma\left(0\right)=\gamma_{I}\,,\\[6pt]
        \gamma\left(t\right)^2=\gamma\left(t\right)\,,\quad Q\left(t\right)=\gamma\left(t\right)-P^0_-\,,
    \end{cases}
\end{equation}
where $P_I$ is the initial datum and $\DD_{Q(t)}$ is as in \eqref{eq:mfdv}, assuming \eqref{eq:n(t)}.
\smallskip

In the next section we address global well-posedness for the Cauchy problem \eqref{eq:main-cauchy-problem}, under suitable assumptions. In the 3D massive case, the analogous system has been studied in \cite{HaiLewSpa-2005-LMP} with an external electrostatic, but also briefly discussing also the time-dependent case (see \cite[Remark~2.2]{HaiLewSpa-2005-LMP}). In \cite{Mor-2025-JDE}, the same system coupled with Newtonian nuclear dynamics has been considered, in order to study relativistic electrons interacting with classical nuclei.
%%%%%%%%%%%%%%%%%%%%%%%%%%%%%%%%%%%%%%%
\subsection{The main result}
We are now in a position to state our main result about the system \eqref{eq:main-cauchy-problem}. Before doing so, let us define the following function space: 
\begin{multline}\label{eq:Y}
    \YY=\{Q\in\B\left(\HH_\Lambda\right)\,: \,Q^*=Q\,,\,\abs{\DD^0}^{1/2}\left(Q^{++}-Q^{--}\right)\abs{\DD^0}^{1/2}\in\mathfrak S^1\,,\\
    Q\abs{\DD^0}^{1/2}\in\mathfrak S^2\,,\rho_Q\in\mathcal C\}\,.
\end{multline}
Notice that the above space is different from \eqref{eq:X}, which is the space introduced in the study of the BDF energy in the previous section. Compared to that case, here the trace-class condition has been modified and it is essentially defined in terms of the kinetic part of the BDF energy. Indeed, it coincides with finiteness of the kinetic term, provided that 
\[
-P^0_-\leq Q\leq P^0_+\,,
\]
since in that case \cite{BacBarHelSie-1999-CMP,HaiLewSpa-2012-JMP} $Q^{++}-Q^{--}\geq Q^2$. Ensuring this property will be crucial in the proof of \cref{stat:main-result}, precisely in \cref{stat:global-existence}. Notice that we also add the requirement that $\rho_Q\in\mathcal C$, which allows to deal with states with finite Coulomb energy, and does not seem to follow from the other properties in the definition \eqref{eq:Y}.

\begin{thm}\label{stat:main-result}
Let $\Lambda>0$, $\nu\in C^1\left(\intco{0}{+\infty},\CC\right)$ and $\gamma_I$ such that $Q_I=\gamma_I-P^0_-\in\YY$. If the Fermi velocity $v_\mathrm{F}$ satisfies
\[
v_\mathrm{F} \geq v_\mathrm{c}=h^{-1}\left(2\right),
\]
the Cauchy problem \eqref{eq:main-cauchy-problem} admits a unique global solution
\[
\gamma\left(t\right)\in C^1\left(\intco{0}{+\infty},P^0_-+\YY\right)\,.
\]
\end{thm}
\begin{rmk}
    Notice that the assumption on the Fermi velocity $v_F$ can be removed when the exchange term is neglected. This can be easily seen from the proof of \cref{stat:energy-boundedness}, where it is only used to exploit Kato's inequality to estimate the exchange term with the kinetic one. Therefore, in the reduced model our result holds for any $v_F>0$.
\end{rmk}
%\begin{rmk}
 %   \textcolor{red}{Compared to the analogous result in the 3D massive case with stationnary electric potentials \cite[Theoreom~1.1]{HaiLewSpa-2005-LMP}, notice that the time dependence of local defects here causes the conservation of energy not to hold along any trajectory. Moreover, also the conservation of charge is violated because of the absence of the spectral gap in the 2D massless Dirac operator appearing in our model.}
%\end{rmk}
The next section is devoted to the proof of \cref{stat:main-result}.
%%%%%%%%%%%%%%%%%%%%%%%%%%%%%%%%%%%%%%%%%%%%%%%%%%%%%%%%%%%%
%%%%%%%%%%%%%%%%%%%%%%%%%%%%%%%%%%%%%%%%%%%%%%%%%%%%%%%%%%%%
\section{Global well-posedness: proof of \texorpdfstring{\cref{stat:main-result}}{main result}}\label{sec:global-well-posedness}
Our proof follows by classical arguments, already used also in Hartree-Fock theory \cite{CanLeB-1999-MMMAS,ChaGla-1975-JMP} and in the massive 3D analogue of the problem \cite{HaiLewSpa-2005-LMP, Mor-2025-JDE}
. We work in the space \eqref{eq:Y}, defined in the previous section, endowed with its natural norm
\begin{multline}\label{eq:norm}
\norm{Q}= \norm{ \abs{\DD^0}^{1/2}\left(Q^{++}-Q^{--}\right)\abs{\DD^0}^{1/2} }_{\mathfrak S^1\left(\HH_\Lambda\right)}
+\norm{\abs{\DD^0}^{1/2}Q}_{\mathfrak S^2\left(\HH_\Lambda\right)}+\norm{\rho_Q}_\CC
\end{multline}
Then, we rewrite the evolution problem in terms of $Q\left(t\right)=\gamma\left(t\right)-P^0_{-}$, as follows
\begin{equation}\label{eq:Q-Cauchy-problem}
    \begin{cases}
        i\frac{d}{dt}Q\left(t\right)=\left[\DD_{Q\left(t\right)}, Q\left(t\right)\right]+\left[V_{Q\left(t\right)},P^0_-\right], \\[6pt]
        Q\left(0\right)=Q_{I}\in \YY,
    \end{cases}
\end{equation}
where
\[
V_{Q}=P_\Lambda\left(\left(\rho_Q-\nu\left(t\right)\right)\ast \frac{1}{|\cdot|}-\frac{Q\left(x,y\right)}{\abs{x-y}}\right)P_\Lambda\,.
\]
Notice that the presence of the projector, which in Fourier space reads as
\[
\widehat{P}_{\Lambda}(p)=\mathds{1}(\vert p\vert\leq\Lambda)\,,
\]
is due to the fact that we work on the Hilbert space $\mathfrak H_\Lambda$, defined in \eqref{eq:HL}.

The first step in our proof consists in proving the existence of a unique local solution in $\YY$ for the Cauchy problem \eqref{eq:Q-Cauchy-problem}. Then, we show that the BDF energy functional is constant along this solution. From this property, together with its coercivity, we easily infer that the solution is in fact global in time.
\begin{rmk}
    Notice that the existence of a coercive conserved energy is a huge advantage of the Bogoliubov-Dirac-Fock theory when compared to the Dirac-Fock model or to non-linear Dirac equations. Indeed, to our knowledge, global well-posedness results for these other models have been proved only under suitable assumptions on the smallness of the initial data (with respect to an appropriate Sobolev norm). See, \emph{e.g.}, the papers \cite{BejHer-2014-CMP,BejHer-2015-CMP,BouCan-2015-IMRN,EscVeg-1997-JMA,MacNakOza-2003-RMI}. This is mainly due to the lack of an a priori energy estimate in these cases. 
\end{rmk}
Let us now define $F$ as follows
\begin{equation}\label{eq:F}
    F\left(Q\right)=\left[\DD_Q, Q\right]+\left[V_Q,P^0_-\right]\,.
\end{equation}
A local well-posedness result for \eqref{eq:Q-Cauchy-problem} reduces to the following lemma.
\begin{lem}\label{stat:F-lipschitz}
    Given $\Lambda>0$, $\nu\in C^1\left(\intco{0}{+\infty},\CC\right)$, the map $F$ in \eqref{eq:F} is well defined as $F:[0,+\infty)\times \YY\to \YY$. Moreover, it is  locally Lipschitz on $\YY$, uniformly on bounded time intervals. As a consequence, for any $\gamma_I$ such that $Q_I=\gamma_I-P^0_{-}\in\YY$, there exists a unique local solution $Q\in C^1\left(\intco{0}{T},\YY\right)$ of the Cauchy problem \eqref{eq:Q-Cauchy-problem}.
\end{lem}
\begin{proof}
Notice that the time dependence on the map $F$ is only due to the external electric potential 
\[
\nu=\nu(t)\,.
\]
We will omit to explicitly keep track of it, for simplicity.

Observe that 
\begin{multline*}
    F\left(Q_1\right)-F\left(Q_2\right)= \left[D_\varphi,Q_1-Q_2\right]+\left[V'_{Q_1-Q_2},P^0_-\right] \\
    + \left[V'_{Q_1}-V'_{Q_2}, Q_1\right]+\left[V'_{Q_2}, Q_1-Q_2\right]\,,
\end{multline*}
where 
\[
D_{\varphi(t)}:=P_\Lambda \left(v_\mathrm{F} D^0+\varphi(t) \right)P_\Lambda \,,\quad \varphi(t)=- \nu(t)\ast\vert\cdot\vert^{-1}\,,
\]
and
\begin{equation}\label{eq:V'Q}
V'_Q:=P_\Lambda\left(\rho_Q\ast\vert\cdot\vert^{-1} - \frac{Q\left(x,y\right)}{\abs{x-y}}\right)P_\Lambda
\end{equation}
We now proceed to estimate the terms in the r.h.s. above.
\smallskip

\noindent
{\sl Step I: estimates for $\left[D_\varphi,Q_1-Q_2\right]$.}

Notice that
\begin{multline}\label{eq:IS1}
    \norm{\abs{\DD^0}^{1/2}\left(\left[D_\varphi, Q\right]^{++}-\left[D_\varphi, Q\right]^{--}\right)\abs{\DD^0}^{1/2}}_{\mathfrak S^1\left(\HH_\Lambda\right)}\\
    \leq\norm{\abs{\DD^0}^{1/2}\left(\left(D_\varphi Q\right)^{++}-\left(D_\varphi Q\right)^{--}\right)\abs{\DD^0}^{1/2}}_{\mathfrak S^1\left(\HH_\Lambda\right)}\\
    +\norm{\abs{\DD^0}^{1/2}\left(\left(Q D_\varphi\right)^{++}-\left(Q D_\varphi\right)^{--}\right)\abs{\DD^0}^{1/2}}_{\mathfrak S^1\left(\HH_\Lambda\right)}
\end{multline}
and that it suffices to control the first term on the r.h.s.,
as the other estimates follow by similar arguments. First of all, notice that $D^0$ and $\varphi$ are both Fourier multipliers, as
\[
\widehat{D}_\varphi\left(p\right)=\mathbbm{1}\left(\vert p\vert\leq\Lambda\right)\left(v_{\mathrm{F}}\sigma\cdot p- \widehat{\nu}\left(p\right)\vert p\vert^{-1} \right) \mathbbm{1}\left(\vert p\vert\leq\Lambda\right)\,.
\]
So, as such, they commute with $P^0_\pm$ and $\abs{\DD^0}^{1/2}$. Then we have
\begin{multline}\label{eq:dvq}
    \norm{\abs{\DD^0}^{1/2}\left(\left(D_\varphi Q\right)^{++}-\left(D_\varphi Q\right)^{--}\right)\abs{\DD^0}^{1/2}}_{\mathfrak S^1}\\
    =\norm{D_\varphi\abs{\DD^0}^{1/2}\left(Q^{++}-Q^{--}\right)\abs{\DD^0}^{1/2}}_{\mathfrak S^1}\\
    \leq \norm{P_\Lambda \left(v_\mathrm{F} D^0\right) P_\Lambda}_{\mathfrak S^\infty}\norm{\abs{\DD^0}^{1/2}\left(Q^{++}-Q^{--}\right)\abs{\DD^0}^{1/2}}_{\mathfrak S^1} \\
    +\norm{P_\Lambda\varphi P_\Lambda}_{\mathfrak S^\infty}\norm{\abs{\DD^0}^{1/2}\left(Q^{++}-Q^{--}\right)\abs{\DD^0}^{1/2}}_{\mathfrak S^1} \,.
\end{multline}
Thanks to the ultraviolet cut-off, $\norm{P_\Lambda D^0 P_\Lambda}_{\mathfrak S^\infty}\leq C\left(\Lambda\right)$. Therefore, it suffices to estimate $\norm{P_\Lambda\varphi P_\Lambda}_{\mathfrak S^\infty}$. Taking $h\in\mathfrak H_\Lambda$ and by the Sobolev embedding, one gets
\[
\begin{split}
\norm{\nu\ast \vert\cdot\vert^{-1}h}_{\mathfrak H_{\Lambda}}&\leq
\norm{\nu\ast \vert\cdot\vert^{-1}}_{L^4}\norm{h}_{L^4} \\
&\leq C \norm{ \vert p\vert^{1/2}\nu\ast \vert\cdot\vert^{-1} }_{L^2}\norm{ \vert p\vert^{1/2}h}_{L^2} \\
&\leq C \Lambda^{1/2}\norm{\nu}_{\CC}\norm{ h}_{L^2} \\
&\leq C \Lambda^{1/2}\sup_{t\in[0,T]}\norm{\nu(t)}_{\CC}\norm{ h}_{L^2}\,,
\end{split}
\]
where the last inequality is due to the presence of the ultraviolet cut-off and to the observation that 
\[
\norm{ \vert p\vert^{1/2}\nu\ast \vert\cdot\vert^{-1}}_{L^2} =\norm{ \vert p\vert^{-1/2}\widehat{\nu} }_{L^2} =D\left(\nu,\nu\right)^{1/2}=\norm{ \nu}_{\CC}\,.
\]
Combining the above results, we get
\begin{multline*}
    \norm{ \abs{\DD^0}^{1/2}\left(\left(D_\varphi Q\right)^{++}-\left(D_\varphi Q\right)^{--}\right)\abs{\DD^0}^{1/2}}_{\mathfrak S^1}\\
    \leq C\left(\Lambda,\sup_{t\in[0,T]}\norm{\nu(t)}_{\CC}\right)\norm{\abs{\DD^0}^{1/2} \left(Q^{++}-Q^{--}\right)\abs{\DD^0}^{1/2}}_{\mathfrak S^1}\,,
\end{multline*}
and then, estimating the last term in \eqref{eq:IS1} in a similar way we find
\begin{equation}\label{eq:lip1s1}
    \norm{\abs{\DD^0}^{1/2}\left[D_\varphi, Q\right]^{++}\vert\DD^0\vert^{1/2}}_{\mathfrak S^1}\leq C\left(\Lambda,\sup_{t\in[0,T]}\norm{\nu(t)}_{\CC}\right)\norm{\abs{\DD^0}^{1/2} \left(Q^{++}-Q^{--}\right)\abs{\DD^0}^{1/2}}_{\mathfrak S^1}\,.
\end{equation}
Similar arguments yield
\begin{equation}\label{eqLlip2s2}
    \norm{\left[D_\varphi, Q\right]\abs{\DD^0}}_{\mathfrak S^2}\leq C\left(\Lambda,\sup_{t\in[0,T]}\norm{\nu(t)}\right)\norm{Q\abs{\DD^0}^{1/2}}_{\mathfrak S^2}\,.
\end{equation}
We now estimate the Coulomb norm. We have 
\begin{multline*}
    \frac{\widehat{\rho}_{\left[D^0,Q\right]}\left(k\right)}{\vert k\vert^{1/2}}= \frac{\vert k\vert^{1/2}}{2\pi \vert k\vert}\int_{\{\vert p\vert\leq \Lambda\}}\tr_{\C^2}\left(\widehat{D^0}\left(p+k/2\right)\widehat{Q}\left(p+k/2,p-k/2\right)\right. \\
    \left.-\widehat{Q}\left(p+k/2,p-k/2\right) \widehat{D^0}\left(p-k/2\right)\right)\,dp \\
    =\frac{\vert k\vert^{1/2}}{2\pi}\int_{\{\vert p\vert\leq\Lambda\}}\tr_{\C^2}\left(\frac{\sigma\cdot k}{\vert k\vert}\widehat{Q}\left(p+k/2,p-k/2\right)\right)\,dp\\
\,
\end{multline*}
and using the Cauchy-Schwarz inequality first for matrices and then for functions, one finds
\[
\frac{\widehat{\rho}_{\left[D^0,Q\right]}\left(k\right)}{\vert k\vert^{1/2}}\leq C\left(\Lambda\right) \vert k\vert^{1/2}\left( \int_{\{\vert p\vert\leq\Lambda\}}\,\vert\widehat{Q}\vert^2\left(p+k/2,p-k/2\right)\,dp \right)^{1/2}\,.
\]
Taking the square and integrating over $k$ gives
\[
\norm{\rho_{\left[D^0,Q\right]}}_{\CC}\leq C_1\left(\Lambda\right) \norm{\vert p\vert^{1/2} Q }_{\mathfrak S^2}\leq C_2\left(\Lambda\right) \norm{\abs{\DD^0}^{1/2} Q}_{\mathfrak S^2}
\]
Notice that the density corresponding to $\left[\varphi,Q\right]$ vanishes, namely
\begin{multline}\label{eq:phidzero}
    \widehat{\rho}_{\left[\varphi,Q\right]}\left(k\right)=\iint_{\R^2\times\R^2}\widehat\varphi_Q\left(p+k/2-s\right) \tr_{\C^2}\left(\widehat{Q}\left(s,p-k/2\right)\right)\,dpds \\
    -\iint_{\R^2\times\R^2}\widehat\varphi_Q\left(s-p+k/2\right) \tr_{\C^2}\left(\widehat{Q}\left(p+k/2,s\right)\right)\,dpds=0\,,
\end{multline}
as it can be seen by a change of variables. 
\smallskip

\noindent
{\sl Step II: estimates for $\left[V'_{Q_1-Q_2},P^0_-\right]$.}

We need to deal with a term of the form $\left[V'_Q,P^0_-\right]$. The commutator simplifies to
\begin{equation}\label{eq:commsimpl}
\left[V'_Q,P^0_-\right]=-\left[P_\Lambda \frac{Q\left(x,y\right)}{\abs{x-y}} P_\Lambda, P^0_-\right]\,,
\end{equation}
since $\left[P_\Lambda \rho_Q\ast\vert\cdot\vert^{-1}P_\Lambda,P^0_-\right]=0$, the involved operators being Fourier multipliers. Moreover, observe that
\[
\begin{cases}
    \left[V'_Q,P^0_-\right]^{++}=P^0_+ V'_Q P^0_-P^0_+-P^0_+P^0_- V'_QP^0_+ =0 \\[6pt]
    \left[V'_Q,P^0_-\right]^{--}=P^0_- V'_Q P^0_- -P^0_-P^0_- V'_QP^0_- =0
\end{cases}
\]
using $(P^0_-)^2=P^0_-$ and $P^0_- P^0_+ = P^0_+ P^0_- =0$. Thus,
\[
 \norm{\abs{\DD^0}^{1/2}\left(\left[V'_Q,P^0_-\right]^{++}-\left[V'_Q,P^0_-\right]^{--} \right)\abs{\DD^0}^{1/2}}_{\mathfrak S^1}=0\,.
\]
By \eqref{eq:commsimpl}, it is true that
\[
\norm{\left[V'_Q,P^0_-\right] \abs{\DD^0}}_{\mathfrak{S}^2}\leq 2 \norm{\DD^0}_{\mathfrak S^\infty}\norm{\frac{Q\left(x,y\right)}{\abs{x-y}}}_{\mathfrak{S}^2\left(\HH_\Lambda\right)}\leq 2 C\left(\Lambda\right)\,\norm{\frac{Q\left(x,y\right)}{\abs{x-y}}}_{\mathfrak{S}^2\left(\HH_\Lambda\right)}\,,
\]
and the Hardy inequality gives
\begin{equation}\label{eq:RQ}
\begin{split}
\norm{\frac{Q\left(x,y\right)}{\abs{x-y}}}^2_{\mathfrak{S}^2\left(\HH_\Lambda\right)}&=\iint\frac{\vert Q\left(x,y\right)\vert^2}{\abs{x-y}^2}\,dxdy \\
&\leq \tr\left(\vert p\vert^2 Q^2\right)\leq \Lambda \tr\left(\vert p\vert Q^2\right) \\
&\leq \frac{\Gamma\left(1/4\right)^2}{4\Gamma\left(3/4\right)^2\left(v_{\mathrm F}+g\left(1\right)\right)}\tr\left(\abs{\DD^0}Q^2\right) \\
&=\frac{\Gamma\left(1/4\right)^2}{4\Gamma\left(3/4\right)^2\left(v_{\mathrm F}+g\left(1\right)\right)}\norm{\abs{\DD^0}^{1/2}Q}_{\mathfrak S^2}\,,
\end{split}
\end{equation}
thanks to the ultraviolet cut-off and to the inequality \eqref{eq:dp}. In order to deal with the estimate for the Coulomb norm, observe that 
\[
\left[V'_Q,P^0_-\right]=\left[P_\Lambda\left(\varphi_Q+R_Q\right)P_\Lambda,P^0_-\right]\,,
\]
where $\varphi_Q= \rho_Q\ast\vert\cdot\vert^{-1}$ and $R_Q$ is the operator with integral kernel $-Q\left(x,y\right)/\abs{x-y}$. Then,
\[
\tr\left(\left[V'_Q,P^0_-\right]\right)=\underbrace{\tr\left(\left[P_\Lambda\varphi_QP_\Lambda,P^0_-\right]\right)}_{=0}+\tr\left(\left[P_\Lambda R_QP_\Lambda,P^0_-\right]\right)=\tr\left(\left[P_\Lambda R_QP_\Lambda,P^0_-\right]\right)\,,
\]
since $\tr_{\C^2}\left(\widehat{P}^0_-\left(p\right)\right)=0$, as Pauli matrices have vanishing trace. We have, using twice the Cauchy-Schwarz inequality for matrices and once for functions,
\begin{equation}\label{eq:RQdens}
    \begin{split}
    \frac{\widehat{\rho}_{\left[P_\Lambda R_QP_\Lambda,P^0_-\right]}\left(k\right)}{\vert k\vert^{1/2}}&=\frac{1}{\vert k\vert^{1/2}}\int_{\{\vert p\vert\leq\Lambda\}} \tr_{\C^2}\left(\left(\widehat P^0_-\left(p+k/2\right)-\widehat P^0_-\left(p-k/2\right) \right)R_Q\left(p+k/2, p-k/2\right)\right)\, dp \\
    &\leq \frac{1}{\vert k\vert^{1/2}}\left(\int_{\{\vert p\vert\leq\Lambda\}}\tr_{\C^2}\left(\left(\widehat P^0_-\left(p+k/2\right)-\widehat P^0_-\left(p-k/2\right)\right)^2 \right) \,dp\right)^{1/2} \\
    & \times \left(\int_{\{\vert p\vert\leq\Lambda\}}\abs{R_Q\left(p+k/2, p-k/2\right)}^2  \,dp\right)^{1/2} \\
    &\leq \frac{C\left(\Lambda\right)}{\vert k\vert^{1/2}}\norm{R_Q}_{\mathfrak S^2} \leq \frac{C\left(\Lambda\right)}{\vert k\vert^{1/2}} \norm{Q\abs{\DD^0}^{1/2}}_{\mathfrak S^2}
   \end{split}
\end{equation}
by \eqref{eq:RQ}. Notice that we have used the fact that
\[
\tr_{\C^2}\left(\left(\widehat P^0_-\left(p+k/2\right)-\widehat P^0_-\left(p-k/2\right)\right)^2\right)=2\tr_{\C^2}\left(\widehat P^0_+\left(p+k/2\right)\widehat P^0_-\left(p-k/2\right) \right)\leq 4\,.
\]
Integrating over $k$ in \eqref{eq:RQdens}, we find
\[
\norm{\widehat{\rho}_{\left[P_\Lambda R_QP_\Lambda,P^0_-\right]}}_{\CC}\leq C\left(\Lambda\right)\norm{Q\abs{\DD^0}^{1/2}}_{\mathfrak S^2}
\]
\smallskip

\noindent
{\sl Step III: estimates for $\left[V'_{Q_1}-V'_{Q_2},Q_1\right]$ and $\left[V'_{Q_2},Q_1-Q_2\right]$.}

We consider a term of the schematic form $\left[V'_Q,\widetilde Q\right]$ and look for an estimate of the form
\[
\norm{\left[V'_Q,\widetilde Q\right]}\leq C\norm{Q}\,\norm{\widetilde Q}\,.
\]
In order to bound
\[
\norm{\left(\left[V'_Q,\widetilde Q\right]^{++}-\left[V'_Q,\widetilde Q\right]^{--}\right)\abs{\DD^0}}_{\mathfrak S_1}
\]
it suffices to deal with
\[
\norm{\left(\left(V'_Q\widetilde Q\right)^{++}-\left(V'_Q\widetilde Q\right)^{--}\right)\abs{\DD^0}}_{\mathfrak S^1}\,,
\]
as the arguments for the term $\norm{\left(\left(\widetilde QV'_Q\right)^{++}-\left(\widetilde QV'_Q\right)^{--}\right)\abs{\DD^0}}_{\mathfrak S^1}$ are similar.\par
We start rewriting
\begin{multline}\label{eq:IIIs1}
    \norm{\left(\left(V'_Q\widetilde Q\right)^{++}-\left(V'_Q\widetilde Q\right)^{--}\right)\abs{\DD^0}}_{\mathfrak S^1}\\
    \leq \norm{\left(\left( \left( \rho_{Q}\ast\vert \cdot\vert^{-1}\right)\widetilde Q\right)^{++}  -\left( \left( \rho_{Q}\ast\vert \cdot\vert^{-1}\right)\widetilde Q \right)^{--} \right)\abs{\DD^0}}_{\mathfrak S^1} \\ 
    +\norm{\left(\left( R_Q\widetilde Q\right)^{++}  -\left( R_Q\widetilde Q \right)^{--} \right)\abs{\DD^0}}_{\mathfrak S^1},,
\end{multline}
where $R_Q$ is the operator with integral kernel $Q(x,y)/\vert x-y\vert$.
The first term on the r.h.s. can be bounded as follows
\begin{multline}
    \norm{\left(\left( \left( \rho_{Q}\ast\vert \cdot\vert^{-1}\right)\widetilde Q\right)^{++}  -\left( \left( \rho_{Q}\ast\vert \cdot\vert^{-1}\right)\widetilde Q \right)^{--} \right)\abs{\DD^0}}_{\mathfrak S^1}\\
    = \norm{ \left( \rho_{Q}\ast\vert \cdot\vert^{-1}\right)  \left(Q^{++}-Q^{--}\right) \abs{\DD^0}}_{\mathfrak S^1} \\
    \leq \norm{P_\Lambda \left( \rho_{Q}\ast\vert \cdot\vert^{-1}\right)P_\Lambda}_{\mathfrak S^\infty}  \norm{\left(Q^{++}-Q^{--}\right) \abs{\DD^0}}_{\mathfrak S^1}\,.
\end{multline}
The argument in \emph{Step I} shows that 
\begin{equation}\label{eq:rhohs}
    \begin{split}
    \norm{P_\Lambda \left( \rho_{Q}\ast\vert \cdot\vert^{-1}\right)P_\Lambda}_{\mathfrak S^\infty}&\leq C\norm{\rho_Q}_{\CC} 
    \end{split}\,,
\end{equation}
so that
\[
\norm{\left(\left( \left( \rho_{Q}\ast\vert \cdot\vert^{-1}\right)\widetilde Q\right)^{++}  -\left( \left( \rho_{Q}\ast\vert \cdot\vert^{-1}\right)\widetilde Q \right)^{--} \right)\abs{\DD^0}}_{\mathfrak S^1}\leq C\norm{Q}\,\norm{\widetilde{Q}}\,.
\]
For the other term in \eqref{eq:IIIs1}, we have
\[
\begin{split}
    \norm{\left(\left( R_Q\widetilde Q\right)^{++}  -\left( R_Q \widetilde Q \right)^{--} \right)\abs{\DD^0}}_{\mathfrak S^1}&\leq \norm{\left(R_Q\widetilde Q\right)^{++} \abs{\DD^0}}_{\mathfrak S_1}+\norm{\left(R_Q\widetilde Q\right)^{--} \abs{\DD^0}}_{\mathfrak S_1} \\
    &\leq 2\norm{\frac{Q\left(x,y\right)}{\abs{x-y}}}_{\mathfrak S^2\left(\HH_\Lambda\right)}\,\norm{\widetilde{Q}\abs{\DD^0}^{1/2}}_{\mathfrak S^2}\,\norm{\abs{\DD^0}^{1/2}}_{\mathfrak S^\infty} \\
    &\leq C \Lambda^{1/2}\norm{Q\abs{\DD^0}^{1/2}}_{\mathfrak S^2} \norm{\widetilde{Q}\abs{\DD^0}^{1/2}}_{\mathfrak S^2}\\
    &\leq C \Lambda^{1/2}\norm{Q}\,\norm{\widetilde{Q}}\,,
\end{split}
\]
by \eqref{eq:RQ} and thanks to the ultraviolet cut-off.

In order to estimate the norm $\norm{\left[V'_{Q},\widetilde{Q}\right] \abs{\DD^0}^{1/2}}_{\mathfrak S^2}$, it suffices to consider the term $\norm{V'_{Q}\widetilde{Q} \abs{\DD^0}^{1/2}}_{\mathfrak S^2}$. We have, by \eqref{eq:rhohs} and \eqref{eq:RQ},
\[
\begin{split}
    \norm{V'_{Q}\widetilde{Q} \abs{\DD^0}^{1/2}}_{\mathfrak S^2}&\leq \norm{V'_Q}_{\mathfrak S^\infty}\,\norm{\widetilde{Q}\abs{\DD^0}^{1/2}}_{\mathfrak{S}^2} \\
    &\leq\left(\norm{\rho_Q\ast\vert\cdot\vert^{-1}}_{\mathfrak S^\infty\left(\mathfrak{H}_\Lambda\right)}+\norm{\frac{Q\left(x,y\right)}{\abs{x-y}}}_{\mathfrak S^\infty\left(\mathfrak{H}_\Lambda\right)} \right)\, \norm{\widetilde{Q}\abs{\DD^0}^{1/2}}_{\mathfrak{S}^2}\\
    &\leq\left(\norm{\rho_Q\ast\vert\cdot\vert^{-1}}_{\mathfrak S^\infty\left(\mathfrak{H}_\Lambda\right)}+\norm{\frac{Q\left(x,y\right)}{\abs{x-y}}}_{\mathfrak S^2\left(\mathfrak{H}_\Lambda\right)} \right)\, \norm{\widetilde{Q}\abs{\DD^0}^{1/2}}_{\mathfrak{S}^2} \\
    &\leq C  \norm{Q \abs{\DD^0}^{1/2}}_{\mathfrak S^2}\,\norm{\widetilde{Q} \abs{\DD^0}^{1/2}}_{\mathfrak S^2} \\
    &\leq C \norm{Q}\,\norm{\widetilde{Q}}\,.
\end{split}
\]
Recalling the definition of $V'_Q$ in \eqref{eq:V'Q}, it is not hard to see that the density corresponding to $\left[V'_Q,Q\right]$ vanishes. Indeed, it is true that $\rho_{\left[\varphi_Q,Q\right]}=0$, as it can be seen arguing as for \eqref{eq:phidzero}. Moreover, the density corresponding to $\left[R_Q,Q\right]$ also vanishes, thanks to the symmetry in $Q$.
\medskip

The claim follows combining \emph{Step I-III}.
\end{proof}
Given the previous Lemma, let us consider the maximal solution to \eqref{eq:Q-Cauchy-problem}
\[
Q\in C^1([0,T_{\mathrm{\max}}, \mathcal Y)
\]
provided by \cref{stat:F-lipschitz}. The next result shows that $\gamma\left(t\right)=Q\left(t\right)+P^0_{-}$ is a projector for all $t\in\intco{0}{T_\mathrm{max}}$, ensuring that
\begin{equation}\label{eq:+Q-}
-P^0_{-}\leq Q\left(t\right)\leq P^0_{+}\,.
\end{equation} This property will be crucial later on.
\begin{lem}\label{stat:P-projection}
    Let $\Lambda>0$ and $Q_I=\gamma_I-P^0_{-}\in\YY$. Let $\nu\in C^1\left(\intco{0}{+\infty},\CC\right)$ and $Q\in C^1\left(\intco{0}{T_\mathrm{max}},\YY\right)$ be the maximal solution given by \cref{stat:F-lipschitz}. Then $\gamma\left(t\right)=Q\left(t\right)+P^0_{-}$ is an orthogonal projection, for any $t\in\intco{0}{T_\mathrm{max}}$.
\end{lem}
\begin{proof}
    Let us define $A\left(t\right)=\gamma\left(t\right)^2-\gamma\left(t\right)$. A direct calculation shows, together with \eqref{eq:main-cauchy-problem}, that it satisfies
    \begin{equation}\label{eq:projection-system}
        \begin{cases}
            i\frac{d}{dt}A\left(t\right)=\left[\mathcal{D}_{Q\left(t\right)}, A\left(t\right)\right],\\[6pt]
            A\left(0\right)=\gamma^2_{I}-\gamma_I=0.
    \end{cases}
    \end{equation}
    Since $\mathcal{D}_{Q\left(t\right)}\in C^1\left(\intco{0}{T_{\mathrm{max}}},\B\left(\HH_\Lambda\right)\right)$, the linear Cauchy problem \eqref{eq:projection-system} has a unique maximal solution in $\B\left(\HH_\Lambda\right)$, obviously given by $A\left(t\right)=0$ for any $t\in\intco{0}{T_{\mathrm{max}}}$. This implies that $\gamma\left(t\right)^2=\gamma\left(t\right)$, for any $t\in\intco{0}{T_{\mathrm{max}}}$, and therefore $\gamma\left(t\right)$ is an orthogonal projection along the trajectory.
\end{proof}
We can now prove that the BDF energy \eqref{eq:bdfen} stays bounded on bounded time intervals, along the maximal solution $Q\in C^1\left(\intco{0}{T_\mathrm{max}},\YY\right)$ of \eqref{eq:Q-Cauchy-problem}. This is enough to show that the maximal time of existence is $T_\mathrm{max}=+\infty$.
\begin{lem}\label{stat:energy-boundedness}
    Let $\Lambda>0$, $Q_I=\gamma_I-P^0_{-}\in\YY$ and $\nu\in C^1\left(\intco{0}{+\infty},\CC\right)$ and $Q\in C^1\left(\intco{0}{T},\YY\right)$ be a local solution given by \cref{stat:F-lipschitz}. If $v_\mathrm{F}$ satisfies
    \begin{equation*}
        v_\mathrm{F} \geq v_\mathrm{c}=h^{-1}\left(2\right),
    \end{equation*}
    then there exists $A_{T,\nu}>0$ (depending on $T$ and $\nu$) such that 
    \[
    \E^{\varphi}_\mathrm{BDF}\left(Q\left(t\right)\right)\leq A_{T,\nu} e^{t},\quad \forall t\in[0,T)\,.
    \]
\end{lem}
\begin{proof}
    We start remarking that $\dot{Q}$ is in $\YY$, since $\dot{Q}=-iF\left(Q\right)$ by \eqref{eq:Q-Cauchy-problem} and $F$ is well defined on $\YY$ by \cref{stat:F-lipschitz}. This implies that
    \[
    \frac{d}{dt}\tr\left(Q\left(t\right)\right)=\tr\left(\dot{Q}\left(t\right)\right).
    \]
    Therefore,
    \[
    \begin{split}
        \frac{d}{dt}\E^{\varphi}_\mathrm{BDF}\left(Q\left(t\right)\right)&=\tr\left(\mathcal{D}^0\dot{Q}\left(t\right)\right)+ D\left(\rho_{Q\left(t\right)}-\nu\left(t\right),\rho_{\dot{Q}\left(t\right)}\right)-\tr\left(\frac{Q\left(x,y\right)}{\abs{x-y}}\dot{Q}\left(t\right)\right)-D\left(\dot{\nu}\left(t\right),\rho_{Q\left(t\right)}\right)\\
        &=\tr\left(\mathcal{D}_{Q\left(t\right)}\dot{Q}\left(t\right)\right)-D\left(\dot{\nu}\left(t\right),\rho_{Q\left(t\right)}\right),
    \end{split}
    \]
    by \cite[Lemma~5]{HaiLewSer-2005-CMP} which can be easily adapted to our case. Replacing $\dot{Q}$ by means of \eqref{eq:Q-Cauchy-problem}, we get
    \[
    \frac{d}{dt}\E^{\varphi}_\mathrm{BDF}\left(Q\left(t\right)\right)=-i\tr\left(\mathcal{D}_{Q\left(t\right)}\left[\mathcal{D}_{Q\left(t\right)},Q\left(t\right)\right]\right)-i\tr\left(\mathcal{D}_{Q\left(t\right)}\left[V_{Q\left(t\right)},P^0_{-}\right]\right)-D\left(\dot{\nu}\left(t\right),\rho_{Q\left(t\right)}\right).
    \]
    The second term on the right-hand side can be rewritten as follows
    \[
    \tr\left(\mathcal{D}_{Q\left(t\right)}\left[V_{Q\left(t\right)},P^0_{-}\right]\right)=\tr\left(\mathcal{D}^0\left[V_{Q\left(t\right)},P^0_{-}\right]\right)+\tr\left(V_{Q\left(t\right)}\left[V_{Q\left(t\right)},P^0_{-}\right]\right).
    \]
    In particular, the first term on the right-hand side above vanishes. Indeed,
    \[
    \tr\left(\mathcal{D}^0\left[V_{Q\left(t\right)},P^0_{-}\right]\right)=\tr\left(\abs{\mathcal{D}^0}^{1/2}\left(\left[V_{Q\left(t\right)},P^0_{-}\right]^{++}-\left[V_{Q\left(t\right)},P^0_{-}\right]^{--}\right)\abs{\mathcal{D}^0}^{1/2}\right)
    \]
    and
    \begin{multline*}
        \tr\left(\abs{\mathcal{D}^0}^{1/2}P^0_{-}\left[V_{Q\left(t\right)},P^0_{-}\right]P^0_{-}\abs{\mathcal{D}^0}^{1/2}\right)=\tr\left(\abs{\mathcal{D}^0}^{1/2}P^0_{-}V_{Q\left(t\right)}P^0_{-}P^0_{-}\abs{\mathcal{D}^0}^{1/2}\right)\\
        -\tr\left(\abs{\mathcal{D}^0}^{1/2}P^0_{-}P^0_{-}V_{Q\left(t\right)}P^0_{-}\abs{\mathcal{D}^0}^{1/2}\right)=0,
    \end{multline*}
    \begin{multline*}
        \tr\left(\abs{\mathcal{D}^0}^{1/2}P^0_{+}\left[V_{Q\left(t\right)},P^0_{-}\right]P^0_{+}\abs{\mathcal{D}^0}^{1/2}\right)=\tr\left(\abs{\mathcal{D}^0}^{1/2}P^0_{+}V_{Q\left(t\right)}P^0_{-}P^0_{+}\abs{\mathcal{D}^0}^{1/2}\right)\\
        -\tr\left(\abs{\mathcal{D}^0}^{1/2}P^0_{+}P^0_{-}V_{Q\left(t\right)}P^0_{+}\abs{\mathcal{D}^0}^{1/2}\right)=0.
    \end{multline*}
    Let us now consider the two remaining traces contributing to the time derivative of the BDF energy, that is
    \begin{equation}\label{eq:energy-remaining-terms}
        \tr\left(\mathcal{D}_{Q\left(t\right)}\left[\mathcal{D}_{Q\left(t\right)},Q\left(t\right)\right]\right)+\tr\left(V_{Q\left(t\right)}\left[V_{Q\left(t\right)},P^0_{-}\right]\right).
    \end{equation}
    An easy calculation yields
    \begin{multline*}
        \tr\left(\mathcal{D}_{Q\left(t\right)}\left[\mathcal{D}_{Q\left(t\right)},Q\left(t\right)\right]\right)\\
        =\tr\left(\mathcal{D}^0\left[\mathcal{D}_{Q\left(t\right)},Q\left(t\right)\right]\right)+\tr\left(V_{Q\left(t\right)}\left[\mathcal{D}^0,Q\left(t\right)\right]\right)\\
        +\tr\left(V_{Q\left(t\right)}\left[V_{Q\left(t\right)},P\left(t\right)\right]\right)-\tr\left(V_{Q\left(t\right)}\left[V_{Q\left(t\right)},P^0_{-}\right]\right)
    \end{multline*}
    and the last term on the right-hand side cancels out the second term in \eqref{eq:energy-remaining-terms}. Thus,
    \[
    \begin{split}
        \frac{d}{dt}\E^{\varphi}_\mathrm{BDF}\left(Q\left(t\right)\right)&=\tr\left(\mathcal{D}^0\left[\mathcal{D}_{Q\left(t\right)},Q\left(t\right)\right]\right)+\tr\left(V_{Q\left(t\right)}\left[\mathcal{D}_{Q\left(t\right)},P\left(t\right)\right]\right)-D\left(\dot{\nu}\left(t\right),\rho_{Q\left(t\right)}\right)\\
        &=\tr\left(\mathcal{D}_{Q\left(t\right)}\left[\mathcal{D}_{Q\left(t\right)},P\left(t\right)\right]\right)-\tr\left(\mathcal{D}^0\left[\mathcal{D}_{Q\left(t\right)},P^0_{-}\right]\right)-D\left(\dot{\nu}\left(t\right),\rho_{Q\left(t\right)}\right),
    \end{split}
    \]
    where the first equality holds since $\left[\mathcal{D}^0,P^0_{-}\right]=0$. Now, the second term can be proved to vanish arguing as at the beginning of this proof. Finally, the ciclicality property of the trace implies
    \[
    \frac{d}{dt}\E^\varphi_\mathrm{BDF}\left(Q\left(t\right)\right)=-D\left(\dot{\nu}\left(t\right),\rho_{Q\left(t\right)}\right)\,.
    \]
     Now, \cref{eq:energy-lower-bound} suggests to consider the following (positive) functional
    \[
    \mathcal{G}\left(t\right)=\mathcal{E}^\varphi_\mathrm{BDF}\left(Q\left(t\right)\right)+\frac{1}{2}\norm{\nu\left(t\right)}^2_\CC,
    \]
    whose time derivative is given by
    \[
    \frac{d}{dt}\mathcal{G}\left(t\right)=D\left(\dot{\nu}\left(t\right),\nu\left(t\right)-\rho_{Q\left(t\right)}\right).
    \]
    Therefore, the Cauchy-Schwarz and the elementary inequalities $ab\leq a^2/2 + b^2/2$ (for $a,b\geq 0$) yield
    \begin{align*}
        \frac{d}{dt}\mathcal{G}\left(t\right)&\leq\norm{\dot{\nu}\left(t\right)}_\CC \norm{\nu\left(t\right)-\rho_{Q\left(t\right)}}_\CC\\
        &\leq \frac{1}{2}\norm{\dot{\nu}\left(t\right)}^2_\CC + \frac{1}{2}\norm{\nu\left(t\right)-\rho_{Q\left(t\right)}}^2_\CC.
    \end{align*}
    Our aim is to apply Gr\"onwall's lemma. While the first term on the right-hand side can be easily controlled using \eqref{eq:n(t)}, one would like to find an estimate for the second term making appear the functional $\mathcal{G}$. Using the definition of $h$ in \eqref{eq:h-function}, we infer that
    \[
    \frac{1}{2}\iint\frac{\abs{Q\left(x,y\right)}^2}{\abs{x-y}}dxdy\leq \frac{h\left(v_\mathrm{F}\right)}{2}\tr\left(\abs{\mathcal{D}^0}Q^2\right),
    \]
    where the last term is finite, as $\abs{\mathcal{D}^0}^{1/2}Q\in\schatten{2}$. Moreover, since $-P^0_{-}\leq Q\left(t\right)\leq P^0_{+}$ for any $t\in\intco{0}{T}$ by \cref{stat:P-projection}, we have (see \cite[Sec 4.1]{HaiLewSpa-2012-JMP}) $\tr\left(\abs{\mathcal{D}^0}Q^2\right)\leq\tr\left(\mathcal{D}^0 Q\right)$ and therefore
    \[
    \frac{1}{2}\iint\frac{\abs{Q\left(x,y\right)}^2}{\abs{x-y}}dxdy\leq \frac{h\left(v_\mathrm{F}\right)}{2}\tr\left(\mathcal{D}^0Q\right),
    \]
    In other words, if $\abs{\mathcal{D}^0}^{1/2}Q^{\pm\pm}\abs{\mathcal{D}^0}^{1/2}\in\schatten{1}$ and if the Fermi velocity $v_\mathrm{F}$ satisfies \eqref{eq:vflb}, the exchange term is controlled by the kinetic energy. This implies that when the Fermi velocity satisfies \eqref{eq:vflb} we have the following inequality
    \begin{multline}\label{eq:G-functional-lower-bound}
        \E^\varphi_\mathrm{BDF}\left(Q\left(t\right)\right)+\frac{1}{2}\norm{\nu\left(t\right)}_\CC^2\geq\left(1-\frac{h\left(v_\mathrm{F}\right)}{2}\right)\tr\left(\mathcal{D}^0 Q\left(t\right)\right)\\
        +\frac{1}{2}\norm{\rho_{Q\left(t\right)}-\nu\left(t\right)}^2_\CC
    \end{multline}
    for all $t\in[0,T)$. Observe that \cite{BacBarHelSie-1999-CMP,HaiLewSpa-2012-JMP} \eqref{eq:+Q-} implies that $Q^{++}\left(t\right)-Q^{--}\left(t\right)\geq Q^2\left(t\right)$, and thus
    \begin{equation}\label{eq:HS-norm-control}
        \tr\left(\mathcal{D}^0 Q\left(t\right)\right)=\norm{\abs{\mathcal{D}^0}^{1/2}\left(Q^{++}\left(t\right)-Q^{--}\left(t\right)\right)\abs{\mathcal{D}^0}^{1/2}}_{\schatten{1}}\geq \norm{\abs{\mathcal{D}^0}^{1/2}Q\left(t\right)}_{\schatten{2}}>0,
    \end{equation}
    when $Q\left(t\right)\in\YY$. Now, inequality~\eqref{eq:G-functional-lower-bound} yields
    \[
    \frac{d}{dt}\mathcal{G}\left(t\right)\leq \frac{1}{2}\norm{\dot{\nu}\left(t\right)}^2_\CC +\mathcal{G}\left(t\right),
    \]
    which, by integrating, implies
    \[
    \mathcal{G}\left(t\right)\leq\underbrace{\mathcal{G}\left(0\right)+\frac{1}{2}\int_0^t\norm{\dot{\nu}\left(s\right)}^2_\CC ds}_{\alpha\left(t\right):=}+\int_0^t\mathcal{G}\left(s\right)ds.
    \]
    Since $\alpha$ is clearly non-decreasing, by Grönwall's lemma, one gets
    \[
    \mathcal{G}\left(t\right)\leq\alpha\left(t\right)e^t,\quad\forall t\in\intco{0}{T}.
    \]
    This concludes the proof.
\end{proof}
The following lemma allows to conclued the proof of \cref{stat:main-result}.
\begin{lem}\label{stat:global-existence}
    Let $\Lambda>0$ and $Q_I=\gamma_I-P^0_{-}\in\YY$. Let $\nu\in C^1\left(\intco{0}{+\infty},\CC\right)$ and $Q\in C^1\left(\intco{0}{T_\mathrm{max}},\YY\right)$ be the maximal solution of \eqref{eq:Q-Cauchy-problem}. If $v_\mathrm{F}$ satisfies \eqref{eq:vflb}, then $T_\mathrm{max}=+\infty$.
\end{lem}
\begin{proof}
As already remarked, we need to prove that the norm \eqref{eq:norm} stays bounded as $t\to T^-_{\mathrm{max}}$.
    By \cref{stat:energy-boundedness}, we know that $\mathcal{E}^\varphi_\mathrm{BDF}\left(Q\left(t\right)\right)$ is bounded on $[0,T_\mathrm{max})$. Inequalities~\eqref{eq:G-functional-lower-bound} and \eqref{eq:HS-norm-control} allow to bound the terms in $\norm{Q\left(t\right)}$ involving Schatten norms of $Q\left(t\right)$ using the BDF energy. Moreover, using the Cauchy-Schwarz inequality for the Coulomb scalar product, from inequality~\eqref{eq:G-functional-lower-bound} we also deduce that
    \[
    \begin{split}
    \frac{1}{2}\Vert\rho_{Q\left(t\right)}\Vert^2_\CC-\Vert \rho_{Q(t)}\Vert_\CC \sup_{t\in[0,T_{\mathrm{max}}]}\Vert \nu(t)\Vert_{\CC}\leq\mathcal{E}^\varphi_\mathrm{BDF}\left(Q\left(t\right)\right)\,,
    \end{split}
    \]
    which easily implies that the Coulomb norm of $\rho_{Q\left(t\right)}$ is bounded as $t\to T^-_{\mathrm{max}}$ as well. Therefore, the maximal time of existence is given by $T_\mathrm{max}=+\infty$, thereby concluding the proof of \cref{stat:main-result}.
\end{proof}

\section{Conclusion}\label{sec:conclusion}
The main contribution of this work is the proof of the global well-posedness of the dynamics of a mean-field model for graphene, in which we consider external time-dependent electric fields, such as those induced by local charge defects in the monolayer of carbon atoms.\par
This is achieved controlling the norm of the solution by means of the energy, which is not conserved since the external potential is non autonomous, but remains bounded over finite time intervals.\par
We believe that this model can be seen as a first step in the study of scattering theory. This will be the object of future investigations.

%%%%%%%%%%%%%%%%%%%%%%%%%%%%%%%%%%%%%%%%%%%%%%%%%%%%%%%%%%%%%%%%%%%%%%%%%%%%%%%%%%%%%%%%%%%%%%%%%%%%%%%%%%%%%%%%%%%%%

\addtocontents{toc}{\protect\setcounter{tocdepth}{0}} % From this point on, only show up to \chapters in the ToC

\section*{Acknowledgements}
U.M. has been supported by the European Union's Horizon 2020 research and innovation programme under the Marie Skłodowska-Curie grant agreement N°945332 and by the Italian Ministry of University and Research (MUR) and Next Generation EU under the PRIN 2022
Project [Prot.2022AKRC5P\_\emph{Interacting quantum systems: topological phenomena and effective theories}]. W. B. has been supported by MUR grant \emph{Dipartimento di Eccellenza 2023-2027} of Dipartimento di Matematica at Politecnico di Milano and by the PRIN 2022 Project [Prot.20225ATSTP\_\emph{Nonlinear dispersive equations in presence of singularities}]  - Financed by the European Union - Next Generation EU, Mission 4 - Component 1 - CUP MASTER E53D2300545006 (CUP D53D23005640001).

\section*{Conflict of interest}
The author declares no conflict of interest.

\section*{Data availability statement}
Data sharing is not applicable to this article as it has no associated data. 

\addtocontents{toc}{\protect\setcounter{tocdepth}{2}} % From this point on, only show up to \subsection in the ToC

%%% BIBLIOGRAPHY %%%

\end{document}